\documentclass{acm_proc_article-sp}
\usepackage{verbatim}  

\newtheorem{prop}{Proposition}
\newtheorem{thm}{Theorem}

\newtheorem{corol}{Corollary}

\newtheorem{definition}{Definition}
\newtheorem{approximation}{Approximation}

\usepackage{bbm}
\usepackage{graphics}

\def\E{\mathbb{E}}
\def\F{\mathcal{F}}
\def\N{\mathbb{N}}
\def\Np{\mathcal{N}}

\def\P{\mathbb{P}}
\def\Po{\mathcal{P}}

\def\R{\mathbb{R}}
\def\cal{\mathcal}
\newcommand\ind[1]{\mathbbm{1}_{\left\{#1\right\}}}
\def\Var{\mathbb{V}\mathrm{ar}}

\graphicspath{{Fig/}}

\def\etal{{\em et al.}}
\def\One{1}
\def\Two{2}
\def\Three{3}
\def\Four{4}

\title{A Queueing System for Modeling a File Sharing Principle}
\numberofauthors{2} %  in this sample file, there are a *total*
\author{
\alignauthor
Florian Simatos and Philippe Robert \\ 
       \affaddr{INRIA, RAP project}\\
       \affaddr{Domaine de Voluceau, Rocquencourt}\\
       \affaddr{78153 Le Chesnay, France}\\
       \email{ \{Florian.Simatos,Philippe.Robert\}@inria.fr}
\alignauthor Fabrice Guillemin 
\\
       \affaddr{Orange Labs}\\
       \affaddr{2, Avenue Pierre Marzin}\\
       \affaddr{22300 Lannion}\\
       \email{Fabrice.Guillemin@orange-ftgroup.com}
}

\begin{document}

\conferenceinfo{SIGMETRICS'08,} {June 2--6, 2008, Annapolis, Maryland, USA.} 

\CopyrightYear{2008}

\crdata{978-1-60558-005-0/08/06} 

\maketitle
\begin{abstract}
We investigate in this paper the performance  of a simple file sharing principle. For this
purpose, we consider a system composed  of $N$ peers becoming active at exponential random
times; the  system is initiated  with only  one server offering  the desired file  and the
other peers after becoming active try to download it. Once the file has been downloaded by
a peer, this  one immediately becomes a server.  To investigate  the transient behavior of
this file  sharing system, we study  the instant when  the system shifts from  a congested
state where  all servers available are  saturated by incoming  demands to a state  where a
growing number  of servers are idle.  In  spite of its apparent  simplicity, this queueing
model (with  a random number of  servers) turns out to  be quite difficult  to analyze.  A
formulation  in terms  of an  urn and  ball model  is proposed  and  corresponding scaling
results are derived. These asymptotic results are then compared against simulations.
\end{abstract}

% A category with the (minimum) three required fields
%A category including the fourth, optional field follows...
%\category{C.2.2}{Computer-Communications Networks}{Network Protocols}[protocol verification]


%\keywords{ACM proceedings, \LaTeX, text tagging} % NOT required for Proceedings

\section{Introduction}

This paper  analyzes the  performance of a  simple file  sharing principle during  a flash
crowd scenario when a popular content  becomes available on a peer-to-peer network.  It is
supposed that a given peer is willing to share a given file with a community of $N$ peers,
which are initially asleep.   An asleep peer becomes active at some  random time, i.e., it
tries  to download  the  file from  a  peer having  the  complete file.  Once  a peer  has
downloaded the file, it immediately becomes  a server from which another peer can download
the file. To simplify the model, we assume that the file is in one piece and not segmented
into chunks; the time needed to download the file from one server is supposed to be random
in order to take into account the diversity of upload capacities of peers.

The goal  of this paper is  to understand how the  network builds up in  this situation as
peers join  the system. In particular,  we are interested  in analyzing the growth  of the
number of available servers in the system.  Note that there are eventually as many servers
as peers since each of them can complete the file download.

In spite  of its  apparent simplicity, the  analysis of  the system is  quite difficult
because we have to  cope with a network comprising a random  number of servers: When peers
complete  their download,  they  become  new servers  so  that the  number  of servers  is
continually increasing.   It is assumed  that an incoming  peer chooses a server  with the
smallest number of queued peers.  Other routing policies are considered at the end of this
paper.

The analysis performed  in this paper substantially differs  from earlier studies appeared
so far in the  technical literature in the sense that we  consider the transient formation
of a  network of peers. Yang  and de Veciana~\cite{Yang06:0} considered  a similar setting
which  they  analyzed  with  results  related  to  branching  processes  to  describe  the
exponential growth of  the number of servers. Our  goal in this paper is  precisely to obtain
more  detailed asymptotics of  this transient  regime.  Except  the paper  by Yang  and de
Veciana~\cite{Yang06:0},  most  of the  papers  published so  far  on  the performance  of
peer-to-peer systems assume that  peers join and leave the system and  that a steady state
regime exists.  The problem is then to  evaluate the impact of some parameters of the file
sharing protocol on the  equilibrium of the system.  Different techniques  can be used to
perform such an analysis, for instance by using a Markovian chain to describe the state of
the system, possibly by using approximation techniques when the state space related to the
number of  peers in the  system is  too large.  See  Ge \etal~\cite{Kurose}. A  fluid flow
analysis  with   an  underlying  Markovian   structure  is  proposed  in   Cl\'evenot  and
Nain~\cite{Nain} in order  to model the Squirrel peer-to-peer caching  system.  In Qiu and
Srikant~\cite{Srikant}, the authors directly use a fluid approximation to study the steady
state of  a peer to peer network,  subsequently complemented by diffusion  variations around
the steady state solutions.  In Massouli\'e and Vojnovi\'{c}~\cite{Massoulie}, the authors study
the  performance  of  a file  sharing  system   via a  stochastic  coupon
replication formulation, a  coupon corresponding to a  chunk of a file.  The  goal of this
study is to understand  the impact of the policy applied by  users for choosing coupons on
the  performance of  the  system. The  system  is studied  in equilibrium  as  in Qiu  and
Srikant~\cite{Srikant}.

The rest of this paper is organized as follows: In Section~\ref{sec:model}, we describe
the system under consideration and some heuristics to study the system are presented. It
turns out that the dynamics of the system can be decomposed in two regimes. In the first
one, there are almost no empty servers and we establish an analogy with a random urn and
ball problem on the real line. By approximating the probability of selecting an urn by its
mean value, we analyze in Section~\ref{urnball} the corresponding deterministic urn and
ball problem. The analysis for the random urn and ball problem is much more complicated to
analyze. The complete analysis is done in \cite{mathpaper} and only the main results are
summarized in Section~\ref{sec:mathr}. In Section~\ref{sec:discussion}, we support  via
simulation the different approximations and heuristics made in this paper to analyze the
file sharing system. Concluding remarks are presented in Section~\ref{conclusion}.

\section{Model description}
\label{sec:model}

\subsection{Problem formulation}

We consider throughout this paper a system composed of $N$ peers interested in downloading
a given file. At the beginning, only one  peer (the initial server) has the file and other
peers are  asleep. When  becoming active, after  an exponentially distributed  duration of
time with parameter $\rho$, a peer tries to  download the file from the server that is the
less loaded in terms of number of  queued peers. In particular, the first peer becoming active
downloads the  file from  the initial  server.  The time  needed to  download the  file is
assumed to be exponentially distributed with mean $1$. 

\paragraph{Exponential distributions}
The hypothesis on the distribution on the duration  of the time for a peer to be active is
quite reasonable: this  is a classical situation when a large  number of {\em independent}
users may access some  network. The assumption on the duration of  the time to download is
not realistic in practice since this quantity is related to the size of the file requested
whose distribution is more likely to be bounded by the maximal size of a chunk. As it will
be  seen, even  within this  simplified setting  (in order  to have  a  nice probabilistic
description  of the  process), mathematical  problems turn  out to  be quite  intricate to
solve. In this respect,  our study could be seen as a first step in the analysis of flash
crowd scenarios. It turns
out that our current investigations in the  general case seem to show that the exponential
distribution does not  have a critical impact  on the qualitative behavior as  long as the
FIFO policy  is used by servers. Mathematically, however, numerous  technical points are
not settled in this case.

We assume that peers requesting the file  from the same server are served according to the
FIFO discipline. Note that, because  of the exponential distribution assumption, this case
is equivalent to  the Processor-Sharing discipline, i.e., when $N$ peers  are present for a
duration  of time  $h$,  each of  them  receives the  amount of  work  $h/N$.  Just  after
completing the file  download, a peer immediately becomes a server  from which other peers
can retrieve  the file.  The  problem of  ``free riders'', i.e.,  peers who do  not become
servers after  service completion, is not discussed here.   As it will be  seen, this feature
does not  change significantly the qualitative  properties of the system.   The problem of
servers who disconnect while they have downloads in progress will not be discussed in this
paper.

It  is  worth noting  that  the  model under  consideration  describes  a ``flash  crowd''
scenario. Indeed,  a peer having a  file accepts to share  it with other peers  and we are
interested in the dynamics  of the sharing process when a large  population of peers tries
to download the  file.  Moreover, since the durations for which  these peers stay inactive
are independent and identically distributed, the flow of arrivals of peers into the system
is  not stationary,  but rather  accumulates at  the beginning  and is then less  and less
intense.  We are hence  interested in the transient regime of the  system. Contrary to the
earlier studies \cite{Kurose,Massoulie,Srikant}, we are not interested in the steady state
regime of the system, where peers continually join and leave the system.

It  is  intuitively  clear  that  there  should  exist  two  different  regimes  for  this
system.  Initially, it  starts congested:  many peers  request the  file, and  only  a few
servers are  available.  Afterward,  the situation is  reversed: there  are a large  number of
servers and only a few requests from the remaining inactive peers.

These  two   regimes  clearly  appear  in   Figure~\ref{fig:empty_servers}  depicting  the
simulation results with  $N = 10^6$ peers and  $\rho = 5/6$. It shows that  before time $T
\approx 7$  time units  (or equivalently mean  download times),  there are almost  no empty
servers, while  after that  time, more  and more servers  are empty  until all  peers have
completed  their  download.  But as long as the input  rate is high, a new
server immediately receives a customer. This is all the more true under the routing policy
considered, since new peers entering the system choose an empty server if any.

\begin{figure}[ht]
\scalebox{0.66}{\includegraphics{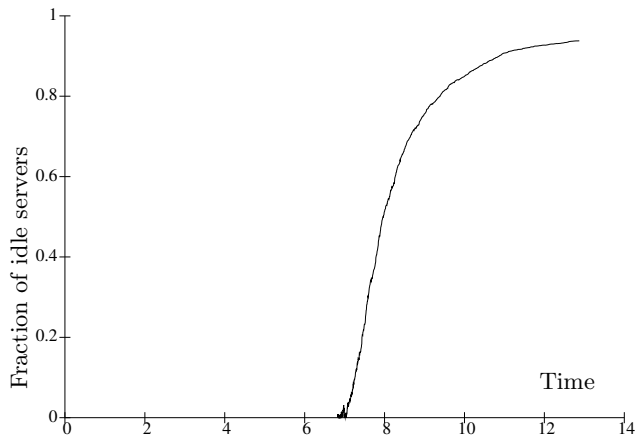}}
\put(-240,20){\rotatebox{90}{Fraction of idle servers}}
\put(-40,20){Time}
\caption{Fraction of idle servers: $N {=} 10^6$ and $\rho {=} 5/6$.}\label{fig:empty_servers}
\end{figure}

\subsection{A Non-Trivial Queueing Model}

From the above  description, the system can  be represented by means of  a queueing system
with a random number of queues.  Initially, the system is composed of a single server, and
once a customer  has completed its service, it  becomes a new server. Since  only a finite
total  number  of  customers is  considered,  there  are  eventually  as many  servers  as
customers.

When peer  inter-arrival times  and file  download times are  assumed to  be exponentially
distributed,  a minimal  Markovian  representation  of this  queueing  model requires  the
knowledge of the number of peers which  are still asleep and the number of peers connected
to each server.  Since this Markov  process is ultimately absorbing (all peers are servers
at the end), the transient behavior of the system is of course the main object of interest
in the analysis. Even in very  simple queueing systems, the transient behavior is delicate
to  analyze  and much  more  difficult  to describe  than  the  stationary behavior.   The
classical $M/M/1$ queue is a good (and  simple) example of such a situation when transient
characteristics  are  not   easy  to  express  with  simple   closed  form  formulas.  See
Asmussen~\cite{Asmussen:01} for example.

Given the multi-dimensional description (with  unbounded dimension) of the Markov process,
the system considered here is much more intricate and challenging. To analyze this system,
a simpler mathematical  model with urns and  balls is used to investigate  the duration of
the  first regime  of this  system.  The  specific  point addressed  in this  paper is  to
describe the transient behavior when $N$ becomes large.

\subsection{Modeling the First Regime}\label{sec:modeling}

Initially,  the input rate is large and therefore a newly  created server 
receives very quickly many requests from the numerous peers becoming  active. The first regime
described  in the previous  section and  illustrated in  Figure~\ref{fig:empty_servers} is
hence characterized  by the fact that  the duration  times during which some  servers are
idle are negligible.   In a  second phase the number of empty servers begins to
be  significant before  increasing very  rapidly in  the last  phase.  This  phenomenon is
discussed in Section~\ref{sec:discussion}. For the first regime, this leads us to describe
the dynamics of the system as follows.

Let $S_n$  be the time  at which the  $n$-th server is  created, with the  convention that
$S_0= 0$  (the initial  server has label  0). During  the $n$-th time  interval $(S_{n-1},
S_n)$ for  $n \geq 1$, there  are by definition exactly  $n$ servers. So if  we assume, as
argued above, that empty servers are negligible during the first regime, $S_n-S_{n-1} $ is
well  approximated by the  minimum of  $n$ independent  exponential random  variables with
parameter $1$.  The random variable $S_n$  can thus be represented  as $S_{n-1}+ E_n^1/n$,
where $E_n^1$  is an  exponential random  variable with parameter  $1$ independent  of the
past. In particular, during the first regime, the following approximation is accurate.

%The duration of time for the first peer to enter the system after time $T_{n-1}$
%is the minimum of $N-n$ exponential random variables with parameter $\rho$ and therefore
%with mean $1/(N-n)\rho$. If this duration is neglected, $S_n-S_{n-1} $ is the minimum of
%$n$ exponential random variables with parameter $1$, $S_n$ can thus be represented as
%$S_{n-1}+ E_n^1/n$, where $E_n^1$ is  exponential random variable with parameter $1$
%independent of the past.

\begin{approximation}\label{approx:heuristic}
For $n\in\N$, as long as the system is still in the first regime, the instant of creation of the $n$-th server is given by $S_n \approx T_n$, where
\begin{equation}\label{eq:T}
	T_n=\sum_{k=1}^n \frac{E_k^1}{k},
\end{equation}
and $E_n^1$ being   i.i.d.  exponential random variables with unit mean.
\end{approximation}

Despite this approximation  seems to be quite rough,  (a rigorous mathematical formulation
of  the   approximation  $S_n   \approx  T_n$  seems   to  be  difficult   to  establish),
Proposition~\ref{prop:heuristic}  and  the  subsequent  discussion  below  provide  strong
arguments to  support its accuracy.  In the  definition of the above  approximation, it is
essential to  determine the duration  of the first  regime, in particular to  know whether
$S_n \approx T_n$ holds or not.

For instance, one could consider as definition for the duration of the first regime the last time when there are no empty servers. This time is
unfortunately not a stopping time and turns out to be much more difficult to study. In
Section~\ref{sec:discussion}, we shall consider different heuristics for evaluating the
length of the first regime. We start the analysis by introducing the index $\nu$ defined
as follows. 

\begin{definition}\label{def:heuristic}
The duration of the first regime is defined as $S_\nu$, where $\nu$ is the first index
$n \geq 1$ so that one or no peer arrive  between $S_{n-1}$ and $S_{n}$.
\end{definition}

According to this definition, the first regime lasts as long as between the creation of
two successive servers, at least two peers arrive in the system. The intuition behind this
heuristic is that, because of the policy for the choice of servers, if many peers arrive
in any interval, then the least loaded servers will receive requests from arriving peers. Thus,
as long as many peers arrive, it is quite rare  for a server to remain  empty. 

The phase transition should occur when the number of arrivals between the creation of two
successive servers is not sufficient to give work to empty servers which are
created. In particular, if no peers arrive in some interval, then there will be at least
two empty servers at the beginning of the next time interval. So the first time when only
a few peers arrive in some interval should be a good indication on the current state 
of the system. A probably more natural heuristic would have been to consider the first
interval in which no peer arrives. Nevertheless, an argument in favor of the former  heuristic is that it  enjoys the following nice
property.

\begin{prop}\label{prop:heuristic}
For $n < \nu$, at most two servers are simultaneously empty in the $n$-th interval $(S_{n-1}, S_n)$.
\end{prop}

\begin{proof}
The proof is by induction. For $n = 1$, the property is trivial, since there is only one
server in the first interval. Consider now $1 < n < \nu$, and suppose that the property
holds for $n-1$. Since at least two peers arrive in the $(n-1)$-th interval, and since
these peers are necessarily routed to empty servers, if any, there is no empty
server just before $S_{n-1}$. Therefore, just after $S_{n-1}$, there are at most two empty servers, and
so the property holds as long as $n < \nu$.
\end{proof}

We are  now able to justify  Approximation~\ref{approx:heuristic}. Indeed, for  $n < \nu$,
the number of non idle servers is  between $n-2$ and $n$. For $n$ large, approximately $n$
servers  are busy,  thus $S_n  - S_{n-1}$  is close  in distribution  to  an exponentially
distributed  random variable with  parameter $n$.   During the  first time  intervals, the
number of empty servers  is negligible. Indeed, consider any finite index  $n$, then it is
easy  to  see that  the  mean number  of  peers  that arrive  in  the  $n$-th interval  is
proportional to $N$. So after the creation  of the $n$-th server, the mean time before the
next arrival behaves  as $1/N$, and so is  very small when $N$ is  large. This intuitively
shows that  the fraction of idle  servers is initially  negligible, which
justifies Approximation~\ref{approx:heuristic}.

From now on, the identification of $S_n$  and $T_n$, where the sequence $(T_n)$ is defined
by Equation~\eqref{eq:T},  is assumed  to hold.  Results on $T_n$  can be  assumed to hold for
$S_n$ when $n < \nu$.

\section{Urn and Ball Problem}
\label{urnball}

Denote by  $(E_i^\rho, 1\leq  i\leq N)$ an i.i.d.\ sequence of  exponentially distributed
random variables with parameter $\rho$. For $i \leq N$, $E_i^\rho$ is the time at which the
$i$-th peer becomes active.

We introduce the  following urn and ball  model on the real line:  The interval $(T_{n-1},
T_n)$ is  the $n$-th urn  and the variables $(E_i^\rho, 1\leq i\leq N)$ are the
locations of $N$ balls thrown on the real line. The set $\{T_{n-1}\leq E_i^\rho\leq T_n\}$
is simply the event  that the $i$-th ball falls into the  $n$-th urn. Conditionally on the
sizes of  the urns, i.e., on ${\cal T}=(T_n)$,  we have that  the probability of  such an
event (which does not depend on $i$) is
\begin{multline}\label{eq:p}
P_n = \P \left(T_{n-1} < E^{\rho}_i < T_n \mid {\cal T} \right) \\ = e^{- \rho T_{n-1}}\left(1 - e^{-\rho E_n^1 / n}\right),
\end{multline}
where the random variables  $E_n^1$, $n \geq 1$, are independent and exponentially distributed with mean unity.

With the above formulation, we have then to deal with the following urn and ball model:
\begin{enumerate}
\item A random probability distribution ${\cal P}=(P_n)$ is given 
(urns with random sizes).
\item $N$ balls are thrown independently according to the probability distribution ${\cal P}$.
\end{enumerate}

It is worth  noting that the above urn and  ball model has an infinite  number of urns. In
addition, although urn  and ball problems have been widely studied  in the literature, our
model  presents a  remarkable feature:  For $i\geq  1$, a  ball falls  into  urn  $i$ with
probability $P_i$ which  is a random variable, but conditionally  on the sequence $(P_n)$,
this is a classical urn and ball problem.  Mathematical results for urn models with random
distributions    are    quite    rare.    See   Kingman~\cite{Kingman:01}    and    Gnedin
\etal~\cite{Gnedin:01}  and the  references therein  where some  related models  have been
investigated.

The random model under consideration will give us some information on the behavior of our
system. The following proposition establishes a simple but important characterization 
for the asymptotic behavior of $(P_n)$. 

\begin{prop}\label{prop:asymptotic}
Let $(E_i^1)$, $i \geq 1$,  be independent exponential random variables with parameter $1$. Then, for $n\in\N$ 
\begin{equation}\label{vougeot}
T_n=\sum_{k=1}^n \frac{E_k^1}{k} \stackrel{\text{dist.}}{=} \max_{1\leq k\leq n} E_k^1,
\end{equation}
and the sequence $(T_n-\log n)$ converges almost surely to a finite random variable
$T_\infty$ whose distribution is given by $\P(T_\infty\leq x)=\exp(-\exp(-x))$ for $x\in\R$.

The conditional probability $P_n$ of throwing  a ball into the $n$-th urn can be written as 
\begin{equation}\label{RD}
P_n = \frac{\rho}{n^{\rho+1}} X_{n-1} Z_n,
\end{equation}
where 
\[
Z_n = \frac{n}{\rho}\left( 1 - e^{-\rho E_n^1 / n} \right) \textrm{ and } X_{n-1} =
n^{\rho} e^{-\rho T_{n-1}}
\]
are independent random variables. As $n$ goes to infinity,
$X_n$ (resp.\ $Z_n$) converges in distribution to $X_\infty$ (resp.\ $Z_\infty$). The 
convergence of $(X_n)$ to $X_{\infty}$ holds almost surely and in $L_q$, for any $q\geq 1$. 

The limiting variable $Z_\infty$ has an exponential distribution with parameter $1$ and
$X_\infty$ has a Weibull distribution with parameter $1/\rho$, 
\begin{equation}\label{eq:lawX}
\P(X_\infty\geq x)=e^{-x^{1/\rho}},\ x\geq 0. 
\end{equation}
\end{prop}

\begin{proof}
Let $E_{(1)}\leq E_{(2)}\leq\cdots\leq E_{(n)}$ be the variables $(E_k^1, 1\leq k\leq n)$ in increasing order. In particular $E_{(n)}=\max_{1\leq k\leq n} E_k^1$. 
With the convention $E_{(0)}{=}0$, due to standard properties of the exponential
distribution, the variables $E_{(i+1)}{-}E_{(i)}$, $i=0$,{\ldots}$,n{-}1$ are independent
and the variable $E_{(i+1)}-E_{(i)}$ is the minimum of $n-i$ exponential variables with
parameter $1$, i.e., has the same distribution as $E_{n{-}i}^1/(n{-}i)$. The distribution
identity~\eqref{vougeot} then follows. 

\begin{comment}
We have $P_n = \rho n^{-\rho-1} X_{n-1} Z_n$, and $X_{n-1}$ and $Z_n$ are
independent random variables. 
\end{comment}

Since $Z_n\stackrel{\text{dist.}}{=} n/\rho( 1 -\exp(-\rho E_1 / n))$, it converges in
distribution to an exponential distribution with parameter one.  

Define 
\[
M_n = \sum_{k=1}^n \frac{E_k^1 - 1}{k} = T_n - H_n,
\]
where $(H_n)$ is the sequence of harmonic numbers, $H_n=1+1/2+\cdots+1/n$.
The sequence $(M_n)$ is clearly a martingale, it is bounded in ${L}_2$ since
\[
\E M_n^2 = \sum_{k=1}^n \frac{\E \left(E_k^1-1\right)^2}{k^2} =  \sum_{k=1}^\infty \frac{1}{k^2} < +\infty.
\]
It therefore converges almost surely. See Williams~\cite{williams91:0} for example. 
The almost sure convergence of $(T_n-\log n)= (M_n +H_n-\log n)$ is thus proved. 
Identity~\eqref{vougeot} gives that, for $x\geq 0$,
\[
\P(T_n-\log n\leq x)=(1-e^{-x-\log n})^n\sim e^{-e^{-x}},
\]
as $n$ goes to infinity. 

Since
\[
X_n=e^{-\rho M_n} e^{\rho(\log(n+1)-H_n)},
\]
one gets the almost sure convergence of $(X_n)$. 
It is easy to check that, for $q\geq 0$,
\begin{align}\label{eq:exp_as}
\E\left(X_n^q\right)&=(n+1)^{q\rho} \prod_{i=1}^n \frac{1}{1 + q\rho / i} \notag \\ &=(n+1)^{q\rho} \frac{\Gamma(n)}{\Gamma(n+q\rho)}\Gamma(q\rho) 
 \sim \Gamma(q\rho),
\end{align}
when $n \to \infty$, where $\Gamma$ is the usual Gamma function, and where the last equivalence easily comes from Stirling's Formula. In particular, for any $q\geq 0$, the
$q$-th moment of $X_n$ is therefore bounded with respect to $n$. One deduces the
convergence in $L_q$ of the sequence $(X_n)$. Since $X_n=\exp(-\rho(T_n-\log(n+1)))$, one
has the equality in distribution $X_\infty=\exp(-T_\infty)$ which gives the law of $X_\infty$. 
\end{proof}

It is important to note that the probability distribution ${\cal P}=(P_n)$ is a {\em random} element in the set of probability
distributions on $\N$. The decay of this distribution follows a power law with parameter
$\rho+1$, because according to the previous proposition, $n^{\rho+1} P_n$ converges in
distribution to $\rho X_\infty Z_\infty$. 
%\[
%P_n\sim \frac{\rho}{n^{\rho+1}} X_{\infty} Z_\infty.
%\]
Using the asymptotic behavior derived in~\eqref{eq:exp_as} with $q = 1$, it is easy to see that the average
probability for a ball to fall into the $n$-th urn satisfies the following relation
\begin{equation}\label{det}
\E(P_n)\sim \frac{\rho\Gamma(\rho)}{n^{\rho+1}}.
\end{equation}
This equivalence suggests the introduction of a deterministic version of the urn and
ball problem considered.

\section{Deterministic Problem}
\label{deterurn}
\subsection{Description}

Denote by ${\cal Q}=(q_n)$ a probability distribution on $\N$ such that
\begin{equation}\label{eqdet}
\lim_{n\to+\infty} n^\delta q_n=\alpha,
\end{equation}
for some $\alpha>0$ and $\delta>1$. For each $n$, $q_n$ can be seen as the probability for
a ball to fall in the  $n$-th urn. When $\delta=\rho+1$ and $\alpha=\rho\Gamma(\rho)$, the
sequence   $(q_n)$  has   the   same  asymptotic   behavior   as  $\E   (P_n)$  given   by
Equation~\eqref{det}. Hence, this model may  be considered as the deterministic equivalent
of the urn and ball problem defined in  the previous section. For the sake of clarity, the
problem with the probability distribution ${\cal  P}$ (resp.\ ${\cal Q}$) will be referred
to as the random (resp. deterministic) problem.

The deterministic problem amounts to throwing $N$ exponential variables with parameter
$\rho$ on the half-real 
line, where this line has been divided into deterministic intervals $(t_{n-1}, t_n)$ with
$t_n = \E T_n$.  The main quantity of interest in the following is the asymptotic behavior with respect to
$N$ of the index of the first urn that does not receive any ball. 

\begin{definition}
Let us denote by $\eta_i^R(N)$ (resp.\ $\eta_i^D(N)$) the number of balls in the $i$-th urn when $N$
balls have been thrown in the random (resp.\ deterministic) urn and ball problem,  and define
\begin{align}
\nu^R(N) = &\inf\{i\geq 1: \eta_i^R(N)=0\}, \label{nu}\\
\nu^D(N) = & \inf\{i\geq 1: \eta_i^D(N)=0\}.\notag
\end{align}
\end{definition}
In view of Definition~\ref{def:heuristic}, to investigate the duration of the first regime
of the system,  the asymptotic behavior of the sequences  $(\nu^R(N))$ and $(\nu^D(N))$ is
analyzed. Since  we consider that the first regime lasts  until one or no peers
arrive  between the  creation  of two  successive  servers, we  should  have to  consider
$\nu'(N)=\inf\{i\geq  1: \eta_i(N) \leq  1\}$ to  be rigorous.  In fact,  the mathematical
analysis of the index of the first empty  urn can easily be extended to the first urn that
receives less than $k$ balls. For the sake of simplicity, we therefore only treat the case
$k=0$. Neither  the orders of  magnitude nor the  asymptotic behaviors established  in the
following are affected by  the value of $k$, and in particular  if we consider $1$ instead
of $0$.

To conclude this section, let us give a rough approximation of the correct order of magnitude for
$\nu^R(N)$ and $\nu^D(N)$ as $N$ gets large. Rigorous mathematical analysis is carried out
in Section~\ref{sec:math}, while Section~\ref{sec:discussion} compares the insights
provided by the two models. 

For $i\geq 1$, $\E(\eta_i^D(N)) = Nq_i \sim \alpha N/i^{\rho+1}$. Hence, in the deterministic
model, a finite number of balls will fall in the $i$-th urn as soon as $i$ is
of the order of $N^{1/(\rho+1)}$ as $N$ becomes large. Hence we expect that in the
deterministic model, $\nu^D(N) / k(N)$ converges in distribution for $k(N) = N^{1/(\rho +
1)}$. Theorem~\ref{thm:determinist} below shows that the location of the first empty urn is in
fact slightly smaller than $N^{1/(\rho + 1)}$, i.e., of the order of $(N / \ln
N)^{1/(\rho+1)}$. Nevertheless this heuristic approach  gives the correct  exponent in $N$. 

Although $\E(\eta_i^R(N))$  has the same asymptotic behavior,  the corresponding heuristic
approach in the case of  the random model is more subtle. Indeed, we have 
$$
\E(\eta_i^R(N)) = N\E(P_i)  \sim N\rho\Gamma(\rho)/i^{\rho+1},
$$
so  the number of  balls falling  in the $i$-th urn should be of the order  $N i^{-\rho-1}$. However, in the random model, the $i$-th interval is with
random length $E^1_i / i$. So from $T_{i-1}$, the next point $T_i$ is at a distance $E_i^1/ i$ and the  first ball is at a distance corresponding to  the minimum of $N i^{-\rho-1}$
i.i.d.\ exponential  random variables with  parameter $1$. Thus, with  this approximation,
the $i$-th interval is empty with probability
\[
\P \left( \frac{E_i^1}{i} \leq \frac{i^{\rho+1}}{N} E_0^1 \right) = 
\frac{1}{1+{N}/{i^{\rho+2}}}.
\]
When  $N \to  \infty$, this  probability is  non negligible  as soon  as $i$  is  of order
$N^{1/(\rho+2)}$, which is significantly below what we found in the deterministic
case. Theorem~\ref{thm:random}  below shows that this  is indeed the  correct answer.  The
order of  magnitude is one order smaller,  compared to the deterministic  case, because of
the variability of the intervals size: to some extent, a very small interval is generated,
so that no balls fall in it, while in the deterministic case, some balls would have.

\subsection{Asymptotic Analysis}

\label{sec:math}
Cs{\'a}ki  and F{\"o}ldes~\cite{Starski} gives the asymptotic behavior
of the distribution  of $\nu^D$ when $N$ is large. % when the sequence $(q_i)$ is  non-increasing.  
A more complete  description of the locations of  the first empty  urns (and not only for the
first one) can however be achieved. For
this purpose, the variable  $W_N^k$ is defined as the number of  empty urns whose index is
less than $k$ when $N$ balls have been thrown. This random variable is formally defined as 
\begin{equation}\label{eqw}
W_N^k=\sum_{i=1}^k I_{N,i}, \text{ with } I_{N,i}=\ind{\eta^D_i(N) = 0}.
\end{equation}
The distribution of $W_N^k$ is
analyzed when $k$ is dependent on $N$. First, some estimates for  the mean  value and the variance of $W_N^k$ are required. 

\begin{prop}\label{propE}
Assume that the sequence $(q_i)$ is non-increasing. For $x>0$, if
\begin{multline}\label{kappa}
\kappa_x(N)= \\  \left\lfloor \left( \alpha\delta \frac{N}{\log N}\right)^{1/\delta}
\left[1+\frac{1+\delta}{\delta}\frac{\log\log N}{\log N}+\frac{\log x}{\log N}\right]\right\rfloor,
\end{multline}
where $\lfloor y\rfloor$ is the integral part of $y>0$, then 
\begin{equation}\label{limW}
\lim_{N\to +\infty} \E\left(W_N^{\kappa_x(N)}\right)=(\alpha\delta)^{1/\delta} x.
\end{equation}
\end{prop}
\begin{proof}
For $k$, $N\in\N$
\begin{equation}\label{EW}
\E\left(W_N^k\right)=\sum_{i=1}^k (1-q_i)^N.
\end{equation}
For $0\leq x\leq 1$,
\[
0\leq e^{-Nx}-(1-x)^N\leq x_N(1-x_N)^{N-1},
\]
where $x_N$ is the unique solution to the equation $\exp(-Nx)=(1-x)^{N-1}$, since the function $x \to e^{-Nx}-(1-x)^N$ has a maximum at point $x_N$. It is 
easily seen that $Nx_N\leq 2$ (in fact $Nx_N\to 2$ as $N\to+\infty$), so that for $N\geq 1$ 
\begin{equation}\label{ineq}
\sup_{0\leq x\leq 1} \left|e^{-Nx}-(1-x)^N\right|\leq \frac{2}{N}.
\end{equation}
With this relation, we obtain 
\[
\left|\E\left(W_N^{k}\right) - \sum_{i=1}^{k} e^{-Nq_i}\right|\leq \frac{2k}{N},
\]
so that for $k=\kappa_x(N)$ and large $N$, $(1-q_i)^N$ can be replaced with  $\exp(-Nq_i)$ in the expression of $\E(W_N^k)$. 

For the sake of simplicity, we  assume that $q_i=\alpha/i^\delta$, for $i\geq 1$.  The general case
of a non-increasing sequence $(q_i)$ follows along the same lines  since the crucial relation below holds true
with a convenient function $q$.  One defines $q(x) = \alpha\min(x^{-\delta},1)$ for $x\geq
0$.
\[
\int_0^{k}e^{-Nq(u)}\, du \leq \sum_{i=1}^{k}e^{-Nq_i} \leq \int_1^{k+1}e^{-Nq(u)}\, du.
\]
The difference between these two integrals is bounded by $2\exp(-\alpha N/k^\delta)$.
Now take $k=k(N)$ with $ k(N)$ with the same order of magnitude as $(N/\log N)^{1/\delta}$, say, $k(N) \sim  A (N/\log N)^{1/\delta}$ for some $A>0$. We have
\[
\E\left(W_N^{k(N)}\right) = \int_1^{k(N)}e^{-Nq(u)}\, du + o(1).
\]
The right hand side of the above equation is given by
\begin{multline}\label{aux1}
\int_1^{k(N)}e^{-\alpha N u^{-\delta}}\, du  \\ = \frac{(\alpha N)^{1/\delta}}{\delta}\int_{\alpha N k(N)^{-\delta}}^{\alpha N}e^{-u}u^{-(\delta+1)/\delta}\, du.
\end{multline}
Now let $H(N)=\alpha N k(N)^{-\delta}$ and consider
\begin{align*}
e^{H(N)}&H(N)^{(1+\delta)/\delta}\int_{H(N)}^{\alpha N}e^{-u}u^{-(\delta+1)/\delta}\, du\\
&=\int_{H(N)}^{\alpha N}e^{-(u-H(N))}\left(\frac{H(N)}{u}\right)^{-(\delta+1)/\delta}\, du\\
&=\int_{1}^{\alpha N/H(N)}H(N)e^{-H(N)(u-1)}\frac{1}{u^{(\delta+1)/\delta}}\, du\\
&\sim \int_{0}^{+\infty}H(N)e^{-H(N)u}\frac{1}{(1+u)^{(\delta+1)/\delta}}\, du\sim 1,
\end{align*}
since $N/H(N)\to+\infty$ and $H(N)\to+\infty$ as $N\to+\infty$. Therefore, an equivalent expression of the integral in
the right hand side of Equation~\eqref{aux1} has been obtained. Gathering these results, we obtain
\begin{multline}\label{detequiv}
\E\left(W_N^{k(N)}\right) = 
\frac{(\alpha N)^{1/\delta}}{\delta}
e^{-H(N)}H(N)^{-(1+\delta)/\delta} + o(1)\\ \sim
\frac{1}{\alpha\delta}\frac{k(N)^{1+\delta}}{N}\exp\left(-\alpha N k(N)^{-\delta}\right).
\end{multline} 
Relation~\eqref{limW} is obtained by taking $k(N)=\kappa_x(N)$.
\end{proof}

The following proposition shows the equivalence of the
variance and the mean value of $W_N^{\kappa_x(N)}$ under a convenient scaling. This result is
crucial to prove the limit theorems of this section.
 
\begin{prop}\label{propvar}
Assume that the sequence $(q_i)$ is non-increasing. For $x>0$, let $\kappa_x$ be defined by Equation~\eqref{kappa}, then
\begin{equation}\label{eq0}
\lim_{N\to+\infty} \left.{\Var\left(W_N^{\kappa_x(N)}\right)}\right/{\E\left(W_N^{\kappa_x(N)}\right)}=1.
\end{equation}
\end{prop}
\begin{proof}
For $k\geq 1$, by using Equation~\eqref{EW} (which does not depend on $\alpha$).
\[
(\E[W_N^k])^2=\sum_{1\leq i, j\leq  k} (1-q_i-q_j+q_iq_j)^N,
\]
and
\[
\E[(W_N^k)^2]=\E[W_N^k]+\sum_{1\leq i\not= j\leq  k} (1-q_i-q_j)^N,
\]
so that, to prove the equivalence of $\Var(W_N^k)$ and $\E(W_N^k)$, it is sufficient  to show
that the quantities
\[
\sum_{1\leq i, j\leq  k} \left[(1-q_i-q_j+q_iq_j)^N-(1-q_i-q_j)^N \right]
\]
and
\[
\sum_{i=1}^{k} (1-2q_i)^N
\]
are negligible with  respect to $\E(W_N^k)$. Since we consider  $k(N) = \kappa_x(N)$, this
amounts to show that  these quantities are $o(1)$ by Proposition~\ref{propvar}. The second term  is the expected number
of   empty  urns  for   the  distribution   $(\tilde{q}_i)$  such   that  $\tilde{q}_i\sim
2\alpha/i^\delta$. Estimate~\eqref{detequiv} shows that
\begin{align*}
\sum_{i=1}^{\kappa_x(N)}(1-2q_i)^N & \sim \frac{1}{2\alpha \delta}\frac{\kappa_x(N)^{1+\delta}}{N} \exp\left(-2\alpha N \kappa_x(N)^{-\delta}\right) \\
&= o\left(\E W_N^{\kappa_x(N)}\right).
\end{align*}

By using the fact that for $a \geq b\geq 0$, $a^N-b^N\leq N(a-b)a^{N-1}$, the second term satisfies
\begin{multline}\label{aux2}
\sum_{1\leq i, j\leq  k} \left[(1-q_i-q_j+q_iq_j)^N-(1-q_i-q_j)^N \right] \\
\leq N \sum_{1\leq i, j\leq  k}q_iq_j (1-q_i-q_j+q_iq_j)^{N-1} \\
= \frac{1}{N} \left(\sum_{i=1}^{k}Nq_i (1-q_i)^{N-1} \right)^2.
\end{multline}
By using a similar method as in the proof of Proposition~\ref{propE}, we obtain the equivalence
\begin{align*}
\sum_{i=1}^{k(N)}Nq_i (1-q_i)^{N-1} &\sim \int_1^{k(N)} Nq(u)e^{-Nq(u)}\, du\\
&\sim (\alpha\delta)^{1/\delta} \alpha x \frac{N}{\kappa_x(N)^{\delta}}= (\alpha\delta)^{1/\delta} x  \log N.
\end{align*}
This equivalence together with Equation~\eqref{aux2} complete  the proof of the
proposition. 
\end{proof} 

\begin{thm}\label{thm:determinist}
Let  $(q_n)$ be  a non-increasing sequence satisfying Relation~\eqref{eqdet}. For $x>0$ and $N\in\N$,  set
\[
\kappa_x(N)= \left \lfloor \left( \alpha\delta \frac{N}{\log N}\right)^{1/\delta}
\left(1+\frac{1+\delta}{\delta}\frac{\log\log N}{\log N}+\frac{\log x}{\log N}\right)\right\rfloor.
\]
When  $N$ goes to infinity, the variable $W_N^{\kappa_x(N)}$ converges in distribution
to a Poisson random variable with parameter $(\alpha\delta)^{1/\delta}x$.

The index $\nu^D(N)$ of the first empty urn defined by Equation~\eqref{nu} is such that
the variable
\begin{equation}\label{eq2}
 \frac{(\log N)^{(1+\delta)/\delta}}{(\alpha\delta N)^{1/\delta}}\nu^D(N)- 
\log N-\frac{1+\delta}{\delta}\log\log N
\end{equation}
converges in distribution to a random variable $Y$ defined by 
$$
\P(Y\geq x)=\exp\left(-(\alpha\delta)^{1/\delta} e^x\right),\quad x\in\R.
$$  
\end{thm}

\begin{proof}
Chen-Stein's  method is  the basic  tool in  the proof  of the  theorem.   See Barbour
\etal~\cite{Barbour:01} for a detailed  presentation of this powerful method.  Let $N$,
and $k$ be  in $\N$ and  $1\leq i_0\leq  k$. The variable $W_N^k$ conditioned on the event
$\{I_{N,i_0}=1\}$ has the same distribution as the number of empty urns when the balls in
the $i_0$-th urn are thrown again until the $i_0$-th urn is empty. It follows that the
number of balls in any other urn is larger than in the case when they are assigned at
first draw. One deduces that for $i\neq i_0$,
\[
\P\left(I_{N,i} =1\mid I_{N,i_0}=1\right) \leq \P\left(I_{N,i} =1\right).
\]
The variables $(I_{N,i},1\leq i\leq k)$ are therefore negatively correlated, see Barbour
\etal~\cite{Barbour:01}. Then, by \cite[Corollary 2.C.2]{Barbour:01}, the
following relation holds,
\begin{multline*}
\sum_{p\geq 0}\left|\P(W_n^k=p)-\frac{\E\left(W_N^k\right)^p}{p!}e^{-\E\left(W_N^k\right)}\right| \\ \leq
1-\left.{\Var\left(W_N^{k}\right)}\right/{\E\left(W_N^{k}\right)}.
\end{multline*}
By taking  $k=\kappa_x(N)$ and by using Propositions~\ref{propE} and~\ref{propvar},
we obtain  the convergence in distribution of $W_N^{\kappa_x(N)}$ to a Poisson distribution with parameter $(\alpha \delta)^{1/\delta}x$.
The last part of the theorem is a simple consequence of the identity  $\P(W_N^k=0)=\P(\nu^D(N)>k)$. 
\end{proof}

The  convergence  in  distribution  of   $\nu^D(N)$  has  been  proved  by  Cs{\'a}ki  and
F{\"o}ldes~\cite{Starski}  with a  different  method.  Our result  gives  a more  accurate
description of the location of empty urns (and not only the first one) near the index $\kappa_x(N)$. 

The following corollary is a straightforward application of the detailed asymptotics
obtained in the above theorem. 
\begin{corol}[Cutoff phenomenon] Under the assumption of Theorem~\ref{thm:determinist}, if 
$$k(N)=\left({N}/{\log N}\right)^{1/\delta},$$
then, as $N$ goes to infinity, the following convergence in distribution holds: For $\beta>0$, 
\[
W_N^{\beta k(N)}\longrightarrow
\begin{cases}
+\infty& \text{ if } \beta >(\alpha\delta)^{1/\delta},\\
0& \text{ if } \beta <(\alpha\delta)^{1/\delta}.
\end{cases}
\]
\end{corol}

So far, only indexes of empty urns have been considered. The result below shows that the first empty urn happens at a time of the order of $\log N$. Remembering the approximation of the peer to peer system, it suggests that the time the system
begins to serve quickly the incoming peers should be of the same order.
 
\begin{corol}[First Empty Urn]\label{corolhit}
Let
\begin{equation}\label{TD}
T^D(N) = T_{\nu^D(N)} =\sum_{k=1}^{\nu^D(N)} \frac{E_k^1}{k}.
\end{equation}
Under the assumptions of  Theorem~\ref{thm:determinist},   the quantity
$$\delta T^D(N)-\log N+ \log\log N-\log(\alpha\delta)$$ converges in distribution to 
$T_\infty$, where $T_\infty$ is the random variable defined in
Proposition~\ref{prop:asymptotic}. 
\end{corol}
\begin{proof}
If $V_N$ is the variable defined by Expression~\eqref{eq2}, then 
\begin{multline*}
\delta \log \nu^D(N) -\log (N)+\log\log N-\log(\alpha\delta)\\= \delta \log\left(\frac{1}{\log N}\left(V_N+\log N+\frac{1+\delta}{\delta}
\log\log N\right)\right).
\end{multline*}
Since by Theorem~\ref{thm:determinist} the sequence $(V_n)$ converges in distribution, it
implies that the right  hand side of the above expression converges in distribution to
$0$. Proposition~\ref{prop:asymptotic} shows that 
\[
E_1^1+\frac{E_2^1}{2}+\cdots+\frac{E_n^1}{n}-\log n
\]
converges almost surely to $T_\infty$.
\end{proof}

\section{Random Problem}\label{sec:mathr}
For the random model,  the probability $P_n$ of selecting the $n$-th urn is given by
Equation~\eqref{RD} of Proposition~\ref{prop:asymptotic}.  In the (almost sure) limit as
$n$ goes to infinity, $X_n\sim X_\infty$ and in distribution, $Z_n$ is asymptotically an
exponentially distributed random variable with parameter $1$. The sequence $(P_n)$ can be
approximated by  $$\left(\frac{\rho}{n^{\rho+1}} X_\infty E_n^1\right),$$ where $(E_n^1)$
are i.i.d.\ exponential variables with unit means. 

In spite of the fact that the decay  of $P_n$ follows a power law, the random factor plays
an important  role. This factor  is composed of  two variables, one (namely $X_\infty$)  is
fixed once for all  and the other (namely $Z_n$)  changes for every  urn. The fact that $Z_n$,
related to  the ``width''  of the  $n$-th urn, can  be arbitrarily  small with  a positive
probability suggests that the index $\nu^R$ of  the first empty urn should be smaller than
the corresponding  quantity for the deterministic case.   This is indeed true but the
situation in  this case  is much  more complex to  analyze. The  complete analysis  of the
random case is given in \cite{mathpaper}, and only sketches of proof are given for Proposition~\ref{prop:process} and Theorem~\ref{thm:random} in the present paper. It  must be noticed that a similar problem where
$X_\infty$  and the  sequence  $(E_n^1,  n \geq  1)$  are independent  is  fairly easy  to
solve. However  here, these random variables  are dependent, and  this dependency requires
quite technical probabilistic tools.

To derive  asymptotic results for $\nu^R$,  as in the previous section, the asymptotic
behavior of the random variable $W_N^k$ defined by 
\[
W_N^k = \sum_{i=1}^k I_{N,i} \textrm{ with } I_{N,i} = \ind{\eta_i^R(N) = 0}.
\]
is investigated. Although in the deterministic case, Chen-Stein's method makes it possible
to reduce the analysis  of $W_N^k$ to its first and second moments,  this is no longer the
case for the random problem.  Indeed, because of the variability of the urns sizes, the random
variables $(I_{N,i}, 1 \leq i \leq k)$ are no longer negatively correlated.  Moreover, the
ratio of the expected value to the variance of $W_N^{k(N)}$ does not converge to $1$ for a
convenient  sequence $(k(N))$  as  in the  deterministic case  (Proposition~\ref{propvar}),
which suggests that if a limit in distribution exists, it cannot be Poisson.

As  was  pointed  out in Hwang and Janson~\cite{Hwang07:0}, the sequence $(NP_i,  1 \leq i \leq k)$ 
plays  a central role in the limiting behavior of $(W_N^k)$. The following technical
proposition gives a result on the asymptotic behavior of this sequence. It is important
since it introduces the scale  $N^{1/(\rho+2)}$ which turns out to be the correct scaling
for the variable $\nu^R(N)$; see \cite{mathpaper} for the proof. 
\begin{prop} \label{prop:process}
Let $x > 0$. When $N$ goes to infinity, the random sequence $(NP_i, 1 \leq i \leq x
N^{1/(\rho+2)})$ converges in distribution to a doubly stochastic Poisson process with a random intensity
$x^{\rho+2}\big(X_\infty\rho(\rho+2)\big)^{-1}$. 
\end{prop}

\begin{proof}
Because of the technicality involved, we only give a sketch of the proof. The reader is referred to~\cite{mathpaper} for more details. 

To prove  the convergence of  the sequence of  point processes $\Np_N  = \sum_{i=1}^{k(N)}
\delta_{\{NP_i\}}$ with $k(N) = xN^{1/(\rho+2)}$, it  is enough to show the convergence of
the   Laplace   transforms  of   these   point   processes   applied  to   some   suitable
functions. Non-negative  continuous functions  with a compact  support would be  enough to
prove the  result, but  the next theorem  requires a  slightly stronger result,  namely it
requires  the  converge  of  Laplace  transforms  for  non-negative  continuous  functions
vanishing at  infinity, i.e.,  that for any  function $f  \geq 0$ continuous  vanishing at
infinity, we have
\[
	\lim_{N \to +\infty} \E\left(e^{-\Np_N(f)}\right) = \E\left(e^{-\Np_\infty(f)}\right)
\]
where, conditionally on $X_\infty$, $\Np_\infty$ is a Poisson process with intensity $x^{\rho+2}X_\infty^{-1} / (\rho+2)$.

The general idea is to condition on the random variable
$X_\infty$. However, for each $n \geq 1$, $X_\infty$ and $Z_n$ are dependent, so that this
cannot be directly done. Instead, the first step of the proof is to show that only the
last terms of the point process matter, i.e., that $\E(e^{-\Np_N(f)})$ and
$\E(e^{-\sum_{\beta(N)}^{k(N)} f(NP_i)})$ have the same limit, for any sequence $\beta(N)
\ll k(N)$. So we are left with large indexes $i \geq \beta(N)$, for which the
approximation $P_i = \rho i^{-\rho-1}X_i Z_i \approx \rho  i^{-\rho-1} X_{\beta(N)} Z_i$
can be justified. The main tool behind this approximation is Doob's Inequality applied to
the reversed martingale $$\overline M_n = \sum_{k \geq n} (E_k-1)/k.$$ And now, due to
this approximation, it is perfectly rigorous to condition on $\F_N = \sigma(E_k, k <
\beta(N))$: since for $i \geq \beta(N)$, $Z_i$ is independent of $X_{\beta(N)}$, we are
exactly left with proving the result for the sequence of point processes $\Np'_N =
\sum_{\beta(N)}^{k(N)} \delta_{\{N x_N i^{-\rho-1} Z_i\}}$ with any converging sequence
$x_N \to x_\infty$ ($x_N$ has to be thought as being equal to $\rho X_{\beta(N)}$). If $f$
has a compact support, it is possible to conclude by applying a result from
Grigelionis~\cite{Grigelionis61:0} to show that this sequence of point processes converges
to a Poisson process with intensity $x^{\rho+2} / (x_\infty (\rho+2))$. In the general
case, the convergence is shown thanks to computations, by controlling the speed at which
$Z_i$ converge in law to an exponential random variable. 
\end{proof}

This result together  with standard poissonization techniques make it possible to prove
the following theorem, which is the main result of this section. 
\begin{thm}\label{thm:random}
Let $\kappa(N) = N^{1/(\rho+2)}$. For $x > 0$, $W_N^{x\kappa(N)}$ converges in
distribution to a Poisson random variable with a random parameter
$x^{\rho+2}\big(X_\infty\rho(\rho+2)\big)^{-1}$ when $N \to \infty$. 
\end{thm}

\begin{proof}
Again, only a sketch of the proof is given. The first step of the proof is to show the result for the random variable $W_{\Po_N}^{x\kappa(N)}$ where $\Po_N$ is a Poisson random variable with parameter $N$, independent of everything else so far. The idea is that the law of $W_{\Po_N}^{x\kappa(N)}$ is not sensitive to the fluctuations of $\Po_N$ around its mean value, equal to $N$, so that the law of $W_{\Po_N}^{x\kappa(N)}$ and of $W_N^{x\kappa(N)}$ will have the same asymptotic behavior.

To show the convergence of $W_{\Po_N}^{x\kappa(N)}$, we consider its generating function: for $u > 0$ and $k \in \N$, we can compute
\[
	\E \left( u^{W_{\Po_n}^k} \right) = \E \left( e^{\sum_{i=1}^k \log\left(1-(1-u)e^{-NP_i}\right)} \right) = \E\left( e^{-\Np_{N,k}(f_u)} \right),
\]
where $\Np_{N,k} = \sum_{i=1}^k \delta_{\{NP_i\}}$, and $f_u(x) = -\log\big(1 - (1-u) e^{-x} \big)$ for $x \geq 0$. Then $\int_0^\infty(1-e^{-f_u}) = 1-u$, so that we conclude with the previous proposition that $W_N^{x\kappa(N)}$ converges to a random variable which, conditionally on $X_\infty$, is a Poisson random variable with parameter $x^{\rho+2} \big(X_\infty \rho (\rho + 2) \big)^{-1}$. The fact that $W_N^{x\kappa(N)}$ and $W_{\Po_N}^{x\kappa(N)}$ have the same asymptotic behavior (in law) then follows by standard arguments.
\end{proof}

This theorem readily yields the following corollary. 

\begin{corol}
The random variable $\nu^R(N) / \kappa(N)$
converges in distribution to a random variable $Y$ such that 
\[
	\P(Y \geq x) = \E \left( e^{-x^{\rho+2}X_\infty^{-1} / (\rho(\rho+2))} \right).
\]
Finally, if $T^R(N) \stackrel{def}{=} T_{\nu^R(N)}$ then, for the convergence in
distribution, 
\begin{equation}
\label{random_time}\lim_{N\to+\infty}\frac{T^R(N)}{\log(N)} {=} \frac{1}{\rho+2}.
\end{equation}
\end{corol}

The fact that the parameter of the limiting Poisson law is random has important effects,
especially concerning the expectation. Indeed, it stems from Equation~\eqref{eq:lawX} and
Proposition~\ref{prop:asymptotic} that $\lim\E (W_N^{x\kappa(N)})$ is proportional to $\E
X_\infty^{-1}$ and $\E X_\infty^{-1} < +\infty$ if and only if $\rho < 1$. Note in
particular that the value $\rho = 1$ plays a special role for our system. 

For $\rho > 1$,
the mean value of $W_N^{x\kappa(N)}$ diverges because it happens that a finite number of
intervals (actually, the $\lfloor \rho \rfloor$ first intervals) capture most of the
balls. This event happens with an increasingly small probability, so that in the limit as
$N$ goes to infinity, it does not have any impact on our system. However, for a fixed $N$,
this event happens with a fixed probability as well. For instance, we commonly observed on
various simulations for $\rho = 2$ and $N = 10000$ that more than $95\%$ of the peers go
to the first server, which is clearly an undesirable behavior of the system.

\section{Discussion} \label{sec:discussion}

In this section, a  set of simulations of the file sharing  principle  is presented to test
the different approximations made  in this paper in term of urn  and ball models. These
simulations  are in  particular used  to justify  Approximation~\ref{approx:heuristic}, as
well as to  compare the insights into the  dynamics of the system provided by  the two urn
and  ball models  studied in  this  paper. Moreover,  another server  selection policy  is
considered, namely when an incoming peer chooses the server at random.

Throughout this section, we discuss the relevance of several random variables. The goal is
to assess the accuracy of the procedure consisting of estimating the length of the first regime by using the random variable  $\nu$  specified in 
Definition~\ref{def:heuristic}. For this purpose, we define different times:
\begin{enumerate}\label{defT2}
\item $\widetilde T_{\One}$ is the first time when two servers are  created and less than 2 peers have arrived.
\item $\widetilde T_{\Two}$ is the last time when there is an empty server.
\item $\widetilde T_{\Three}$ is the first time when the input rate is smaller than the
output rate (see Section~\ref{sec:other_heuristics}). 
\item $\widetilde T_{\Four}$ the first time when a server becomes empty, i.e., when a peer
leaves a server where it was alone. 
\end{enumerate}
We  consider  the  corresponding  quantities   $\widetilde  \nu_i$:  for  $i  =  1,2,3,4$,
$\widetilde  \nu_i$ is  the index  of the  interval $(S_{i-1},  S_i)$ in  which  the event
corresponding  to  $\widetilde  T_i$   happens.  In  particular,  $\widetilde  \nu_{\One}$
corresponds to  Definition~\ref{def:heuristic}. In every  simulation, the averages  of the
quantities $\widetilde \nu_i$ and $\widetilde T_i$ are calculated for the value $\rho = 2$
over $10^4$ iterations of the system which proved to be sufficient in term of numerical
stability. The number of peers $N$ ranges up to $5.10^7$. 

\subsection{Validation of Approximation~\protect\ref{approx:heuristic}}

Definition~\ref{def:heuristic} specifies  the variable considered throughout this paper
to determine the duration of the first regime of the file sharing system. This
variable was chosen for two reasons. First, it is a good indicator of the current
equilibrium  of the system: the output rate begins to be comparable with the input rate
when only a few peers arrive between the creation of  two successive servers. Moreover, the
stopping time defined in this way is mathematically tractable when transposed into  the
context of a certain urn and ball model. Compared to~\cite{Yang06:0}, it is interesting to
note that we are actually able to rigorously prove results, and not only rely on
simulation. As a byproduct, the mathematical problems arising in this context  are
interesting in themselves. 

For the sake of completeness, several points need to be addressed. First, for how long is
Approximation~\ref{approx:heuristic} valid? Since the random variable $\nu$ specified  in
Definition~\ref{def:heuristic} corresponds to $\widetilde \nu_{\One}$, we argued in
Section~\ref{sec:modeling} that this approximation holds until $\widetilde T_{\One}$. This
is the main assumption that makes it possible to cast our problem in terms of urns and
balls, and to derive precise results on $\widetilde \nu_{\One}$ and $\widetilde
T_{\One}$. 

In order to validate the results of Section~\ref{sec:mathr}, we check that $\widetilde
\E(\nu_{\One})$ and $\E(\widetilde T_{\One})$ behave as predicted by
Theorem~\ref{thm:random}. From this theorem, we expect to have $\E(\widetilde \nu_{\One})
\approx A_1 N^{1/(\rho+2)}$ for some constant $A_1$, and $\E(\widetilde T_{\One} )\approx \log
(N) / (\rho + 2)$. Figure~\ref{fig:validation} shows the graphs $\log(\E(\widetilde
\nu_{\One}))$ and $\E(\widetilde T_{\One})$ versus $\log(N)$: the straight lines depicted prove
a good agreement with the theory. Moreover, via  a fitting procedure, one can compute the
slopes of these lines: the results are summarized in Table~\ref{Tab}.  

\begin{figure}[ht]
\begin{center}
\scalebox{0.66}{\includegraphics{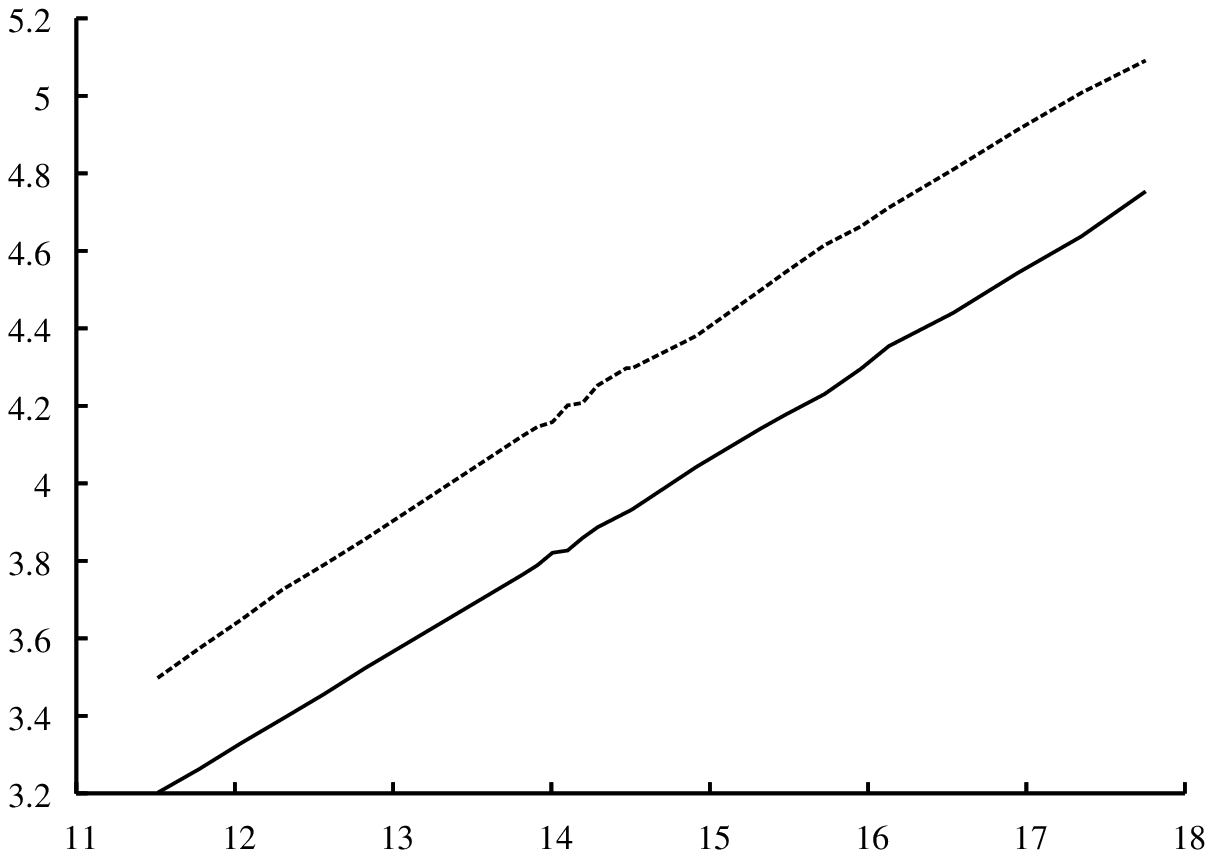}}
\put(-30,20){$\log N$}
\caption{$\log(\E(\widetilde \nu_{\One}))$ (solid) and $\E(\widetilde T_{\One})$ (dashed),
  $\rho=2$.}\label{fig:validation}
\end{center}
\end{figure}

The values of interest in Table~\ref{Tab} are  in the row labeled ``Min'':  we  see  that
simulations  exhibit  a  slope  of  $0.248$  for $\widetilde  \nu_{\One}$ and  of $0.256$
for  $\widetilde T_{\One}$,  whereas the  theory predicts  $0.25$ in  both cases  (because
$\rho  =  2$). These  results are  in   good agreement with
Approximation~\ref{approx:heuristic}, which  justifies the fact that we can 
use this approximation up to time $\widetilde T_{\One}$.

\begin{table}[ht]
\caption{{Coefficients of  growth rates of Fig.~\ref{fig:validation}
    and~\ref{fig:hindsight} in the case $\rho=2$}\label{Tab}}
\bigskip
\scalebox{0.85}{
\begin{tabular}{|l|l|l|l|l|l|l|}\hline
%Policy & $\widetilde \nu_{\One}$& $\widetilde T_{\One}$& $\widetilde \nu_{\Two}$& $\widetilde T_{\Two}$ & $\widetilde \nu_{\Three}$ &$\widetilde T_{\Three}$\\ \hline
%Min & 0.24785 & 0.256581 & 0.376553 & 0.514675& 0.314918& 0.328729\\ \hline
%Random & 0.247057 & 0.257564  & 0.371195 & 0.507898& 0.238315 & 0.253033\\ \hline
Policy\phantom{$\displaystyle {\Sigma^1}^2$} & $\widetilde \nu_{\One}$& $\widetilde T_{\One}$& $\widetilde \nu_{\Two}$& $\widetilde T_{\Two}$ & $\widetilde \nu_{\Four}$ &$\widetilde T_{\Four}$\\ \hline
Min & 0.2478 & 0.2565 & 0.3765 & 0.5146 & 0.3149 & 0.3287\\ \hline
Random & 0.2470 & 0.2575  & 0.3711 & 0.5078 & 0.2383 & 0.2530\\ \hline
\end{tabular}}
\end{table}

\subsection{Accuracy of Urn and Ball Models}
In this section, we compare the random and deterministic urn and ball models with
$\E(\widetilde  T_{\Two})$, the expected value of the last time
when there is  an empty server.  It clearly appears in
Figure~\ref{fig:empty_servers} that $\widetilde T_{2}$ closely corresponds to the shift in
equilibrium   of   the   system,  and   this   fact   has   been  observed   in   numerous
simulations.      However,     as      we     will      see     in      the     following,
Approximation~\ref{approx:heuristic} does not hold until time $\widetilde T_{\Two}$, which
explains why it  is very challenging from a  mathematical point of view to  derive results on
$\widetilde T_{\Two}$. (Note in addition that  $\widetilde T_{\Two}$ is not a stopping time).

Figure~\ref{fig:hindsight} shows that $\widetilde T_{\One}$ is much smaller than
$\widetilde T_{\Two}$: This result is nevertheless  not surprising. Indeed, as discussed
in Section~\ref{sec:model}, results obtained for the random model point out a local
behavior: the first empty urn arrives in a region,  where still many peers arrive in each
interval. Although many peers should arrive in this time interval, this is in reality not
the case  because a very small interval is generated. Thus, in some sense, the order of
magnitude $N^{1/(\rho+2)}$ provided by the random urn and ball model is misleading for the
initial system. 

%\begin{figure}[ht]
%\begin{center}
%\rotatebox{-90}{\scalebox{0.2}{\includegraphics{Min_Index}}}
%\caption{The indices $\widetilde \nu_{1} \leq \widetilde \nu_{\Three} \leq \widetilde \nu_{\Two}$}\label{fig:}
%\end{center}
%\end{figure}
\begin{center}
\begin{figure}[ht]
\scalebox{0.6}{\includegraphics{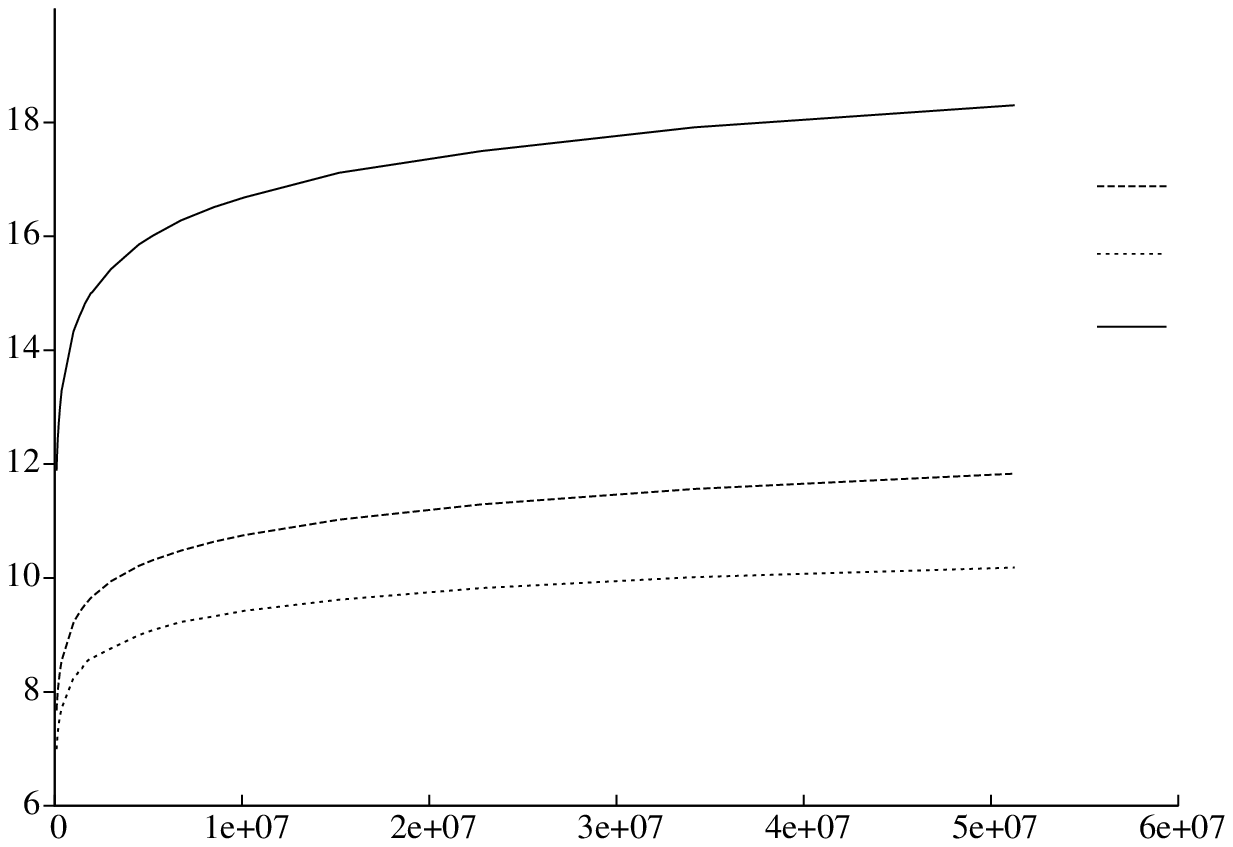}}
\put(-45,115){$\scriptstyle \E(\widetilde T_{1}) $}
\put(-48,101){$\scriptstyle \E( T^{D}) $}
\put(-45,90){$\scriptstyle \E(\widetilde T_{2}) $}
\put(-220,20){\rotatebox{90}{Time}}
\put(-10,20){$N$}
\caption{The times $ \E(\widetilde T_{1})  \leq \E(\widetilde T_{\Two})$ and  $\E(T^D)$
  when $\rho=2$.}\label{fig:hindsight}
\end{figure}
\end{center}

In the deterministic model, the sizes of urns are not random, and the  stochastic fluctuations
arising in the random model do not occur. The deterministic model smooths the local
behavior that appears in the random model, and the order of magnitude $(N / \log
N)^{1/(\rho+1)}$  gives more insight into the global situation of the system. When
only a few peers arrive in an interval, it really means that the equilibrium begins to
shift. One can check in Figure~\ref{fig:hindsight} that the theoretical result
$T^D$ defined by Equation~\eqref{TD} predicted by the deterministic model is closer to
to $\widetilde T_{\Two}$ than to $\widetilde T_{\One}$. 

Although  considering   the  deterministic   model  indeed  improves   the  approximation,
$\widetilde T_{\Two}$ still seems much larger  that $T^D$. However, thanks to our urn and
ball models, we know that the first  order approximation of the times $\widetilde T_i$ is
logarithmic, whereas the  first order approximation for the  indexes $\widetilde \nu_i$ is
polynomial.  Table~\ref{Tab}  provides useful  information  to  understand the  situation.

First, the deterministic model yields a reasonable estimate of the exponent in $\widetilde
\nu_{\Two}$: simulations give  $0.376$ and the deterministic model  $0.333$. Note that the
random model  predicts $0.25$, so a  substantial improvement in accuracy  is obtained when
using  the deterministic  model. This  suggests  that Approximation~\ref{approx:heuristic}
holds  until  $T^D$,  i.e., up to   times  of  order  $N^{1/(\rho+1)}$.   

Second,  we  observe  a
significant  discrepancy  between  the   exponent  for  $\widetilde  \nu_{\Two}$  and  the
coefficient of $\widetilde T_{\Two}$: If Approximation~\ref{approx:heuristic} were to hold
until   $\widetilde    T_{\Two}$,   one   would   have    $\widetilde   T_{\Two}   \approx
\sum_1^{\widetilde  \nu_{\Two}}  E_k^1  /  k$,  which  would  yield,  because  $\widetilde
\nu_{\Two} \approx N^{0.38}$, that $\widetilde T_{\Two} \approx 0.38 \log(N)$. However, we
find that the time $\widetilde T_{\Two}$  is better approximated by $0.52 \log(N)$, and so
Approximation~\ref{approx:heuristic} does not hold  until time $\widetilde T_{\Two}$. This
clearly poses the challenge to  derive  asymptotic results for  $\widetilde T_{\Two}$. Moreover,
this     triggers    another     interesting     question:    For     how    long     does
Approximation~\ref{approx:heuristic} hold?  We give a  partial answer to this  question by
considering  the times  $\widetilde T_{\Three}$  and  $\widetilde T_{\Four}$  in the  next
section.

\subsection{On the Duration of Approximation~1}\label{sec:other_heuristics}
Throughout this  paper, we have  tried to  estimate the time  when the equilibrium  of the
system begins to  shift.   As long as Approximation~1 holds, the input rate $i(t)$  of the
system is the number  of peers, that are not active at  time $t$, times $\rho$, while the
output rate  $o(t)$ is  just the number of servers at  time $t$ (since the service  has 
mean one). Initially, $i(0) =  \rho N$ and 
$o(0) =  1$, and $i(\infty)  = 0$  and $o(\infty) =  N$.  To study  the time at  which the
equilibrium of the  system begins to shift,  it is therefore very natural  to consider the
first time $\widetilde T_{\Three}$ at which $i(t) < o(t)$, i.e., when the number of servers
is greater than  $\rho$ times the number  of non-active peers. As shown  in the following,
considering this  time leads to  the order of  magnitude given by the  deterministic model
(with less precise asymptotics of course).

For times $t < \widetilde T_{\Three}$, we assume that Approximation~\ref{approx:heuristic}
holds, so that we can cast $\widetilde \nu_{\Three}$ in terms of our urn and ball
problem. Let $Z_N^x$ be the number of balls that fall in the $x$ first intervals: 
\[
Z_N^x = \sum_{i=1}^x \eta_i(N) = \sum_{i=1}^N \ind{E_i^\rho \leq T_x}.
\]
The index $\nu_3$ then corresponds to 
$$
\nu_3 = \inf\left\{x: N-Z_N^x \stackrel{def}{=} \widetilde Z_N^x < \frac{x}{\rho}\right\}.
$$
The asymptotic behavior of $\E\big( \widetilde Z_N^x \big)$ when $x$ goes to infinity with $N$ is easy to derive:
\[
\E\big( \widetilde Z_N^x \big) =  N \sum_{i > x+1} \E P_i \sim \alpha N \sum_{i > x+1} i^{-\rho-1} \sim \frac{\alpha}{\rho} N x^{-\rho}
\]
Therefore $\E\big( \widetilde Z_N^x \big) \approx x$ for $x \approx N^{1/(\rho+1)}$, i.e.\
$\widetilde \nu_{\Three}$ is of order $N^{1/(\rho+1)}$, which is the same order of
magnitude as in the deterministic model. Rigorous mathematical analysis could be done to
prove this result, but in our view, considering $\widetilde T_{\One}$ has one main
advantage: Proposition~\ref{prop:heuristic} is almost a rigorous justification of
Approximation~\ref{approx:heuristic}. When considering another time, in particular
$\widetilde T_{\Three}$, we were not able to provide such a strong justification. And as
we have seen in the case of $\widetilde T_{\Two}$, Approximation~\ref{approx:heuristic}
does not hold for the whole first regime, and a strong justification as
Proposition~\ref{prop:heuristic} is therefore very valuable. 

Finally, let us give  some brief results on $\widetilde T_{\Four}$, the  first time when a
server empties. Simulations show  that $\widetilde \nu_{\Four}$ and $\widetilde T_{\Four}$
have    similar    behavior   as    before    (polynomial    and   logarithmic    growths,
respectively). Results in  Table~\ref{Tab} show that the slope  for $\widetilde T_{\Four}$
is   similar   to   the    exponent   of   $\widetilde   \nu_{\Four}$,   suggesting   that
Approximation~\ref{approx:heuristic} holds until $\widetilde T_{\Four}$.

In conclusion, Approximation~\ref{approx:heuristic} holds at least until $N^{1/(\rho+1)}$,
which corresponds to  $\widetilde T_{\One}$ and $\widetilde T_{\Three}$.  However, it does
not hold  until $\widetilde  T_{\Two}$, whereas Figure~\ref{fig:empty_servers}  shows that
until $\widetilde T_{\Two}$,  the system is still in the first  regime. For the particular
value $\rho = 2$, we have $\nu^D \approx A N^{0.33}$ and simulations show that $\widetilde
\nu_{\Two} \approx A_\Two  N^{0.38}$, and so our  approximation by the means of  a urn and
ball    problem   is    not   so    far   from    the   exponent    that   we    want   to
capture.  Proposition~\ref{prop:heuristic} shows that  until $\nu^D$,  there are  only few
empty servers: so between $T^D$ and $\widetilde T_{\Two}$, it could happen that there is a
fraction of empty  servers, and although this fraction  is small, it has an  impact on the
system. Similar phenomenon have been observed in Sanghavi \etal~\cite{Sanghavi:0}.

%\subsection{Comparison of routing policies}\label{rout}

To conclude this section, we discuss a different routing policy. Throughout this paper, we
have considered the policy where an incoming peer selects the least loaded server, in
terms of number of peers. This policy is compared against  the random one, where an
incoming peer selects  a server uniformly at random among all possible servers. 

Simulations  show that these  policies are  very close  as shown  in Figures~\ref{fig:23},
\ref{fig:25} and~\ref{fig:24}. The only noticeable difference is concerning $\E(\widetilde
\nu_{\Three})$,  cf.\  Figure~\ref{fig:26}.   However,  Table~\ref{Tab}  shows  that  the
exponents  of $\E(\widetilde  \nu_{\Three})$ are  very similar  in the  random and  in the
minimum policy. One can easily check that they are indeed proportional one to  another.

\begin{figure}[ht]
\begin{center}
\rotatebox{-90}{\scalebox{0.25}{\includegraphics{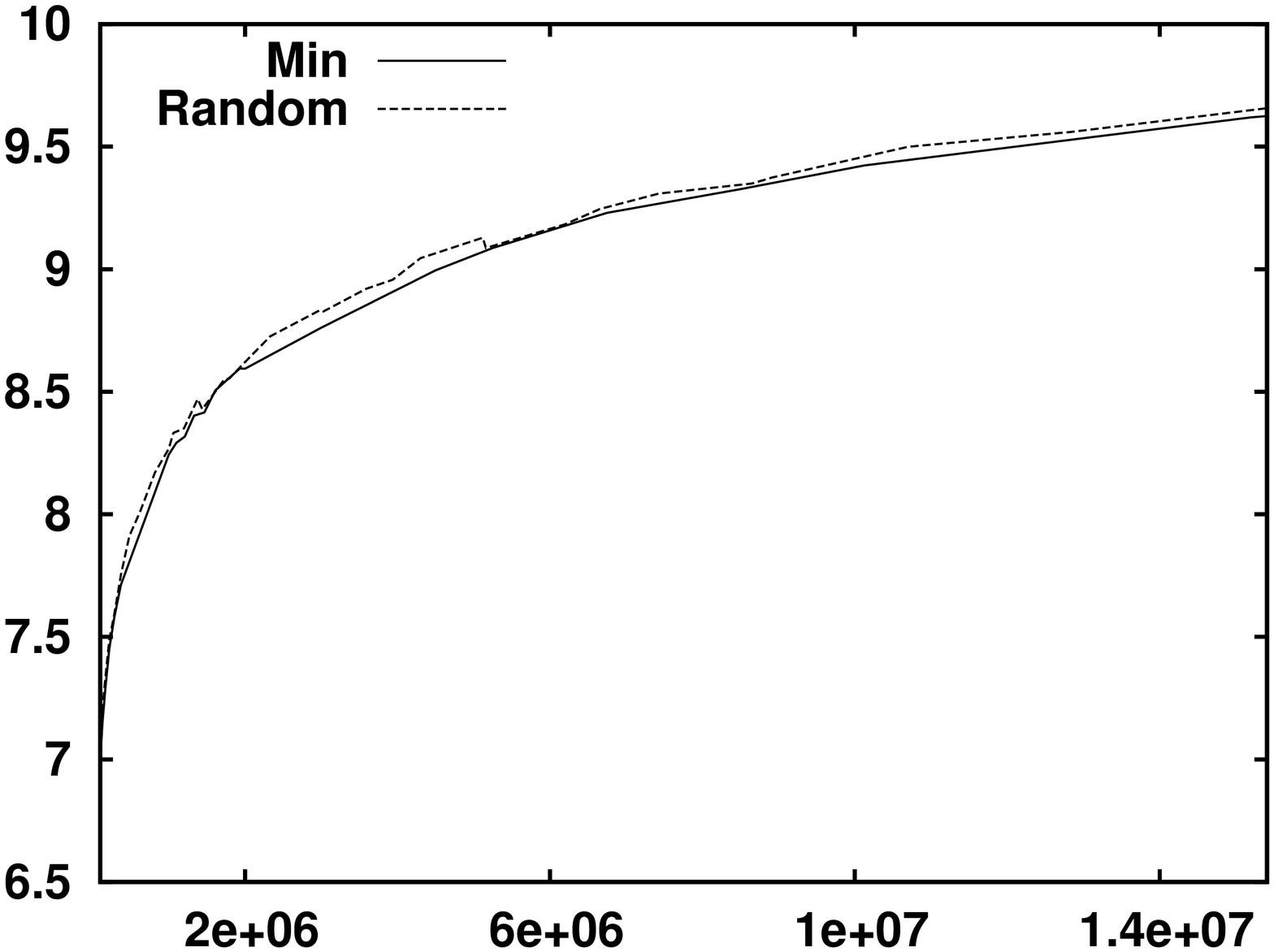}}}
\put(-190,-100){\rotatebox{90}{Time}}
\put(-30,-110){$N$}
\caption{Comparison of Min and Random for $\E(\widetilde T_{\One})$   when $\rho=2$}\label{fig:23}
\end{center}
\end{figure}

\begin{figure}[ht]
\begin{center}
\rotatebox{-90}{\scalebox{0.25}{\includegraphics{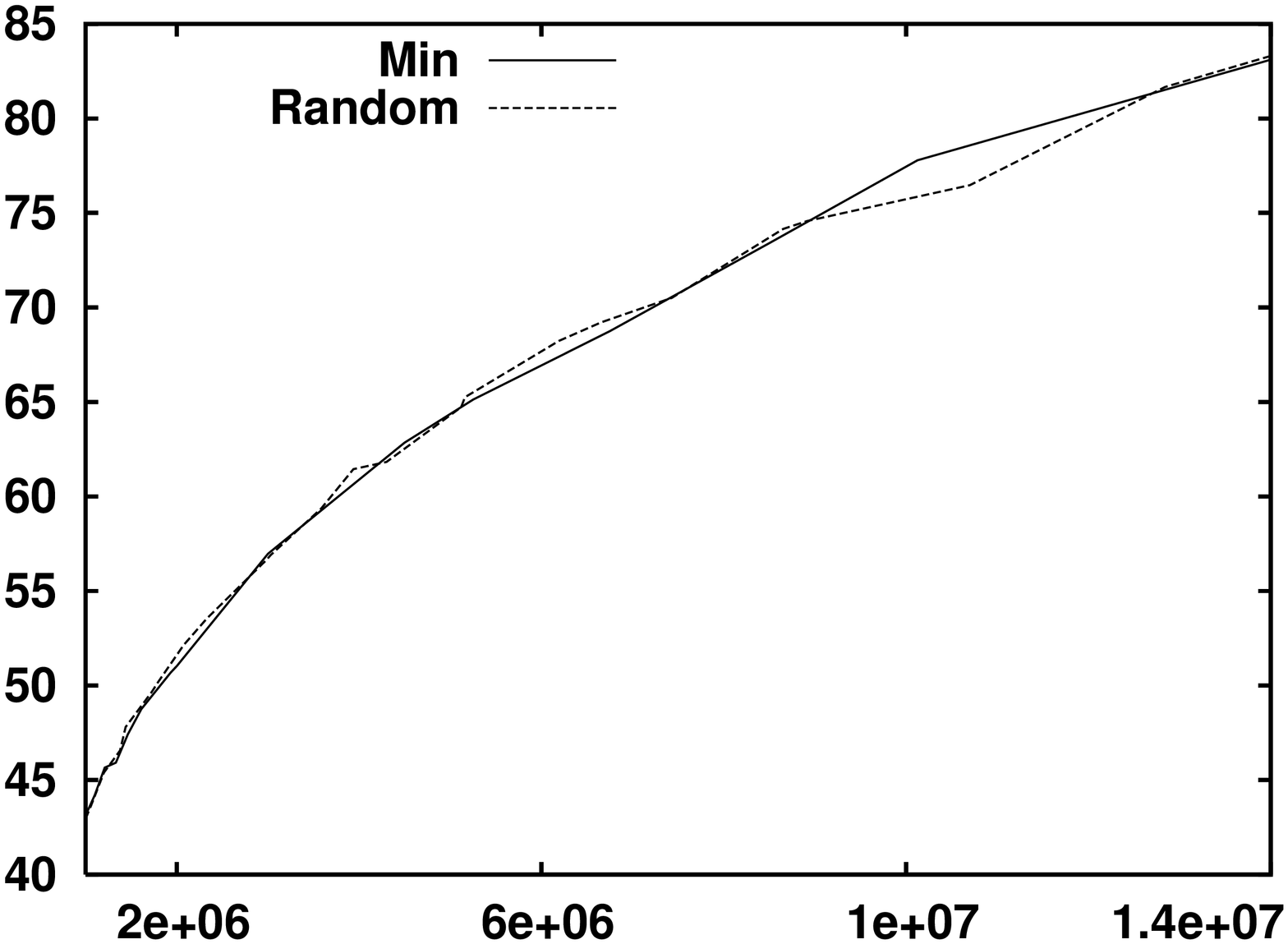}}}
\put(-190,-100){\rotatebox{90}{Time}}
\put(-30,-110){$N$}
\caption{Comparison of Min and Random for $\E(\widetilde \nu_{\One})$   when $\rho=2$}\label{fig:25}
\end{center}
\end{figure}

\begin{figure}[ht]
\begin{center}
\rotatebox{-90}{\scalebox{0.3}{\includegraphics{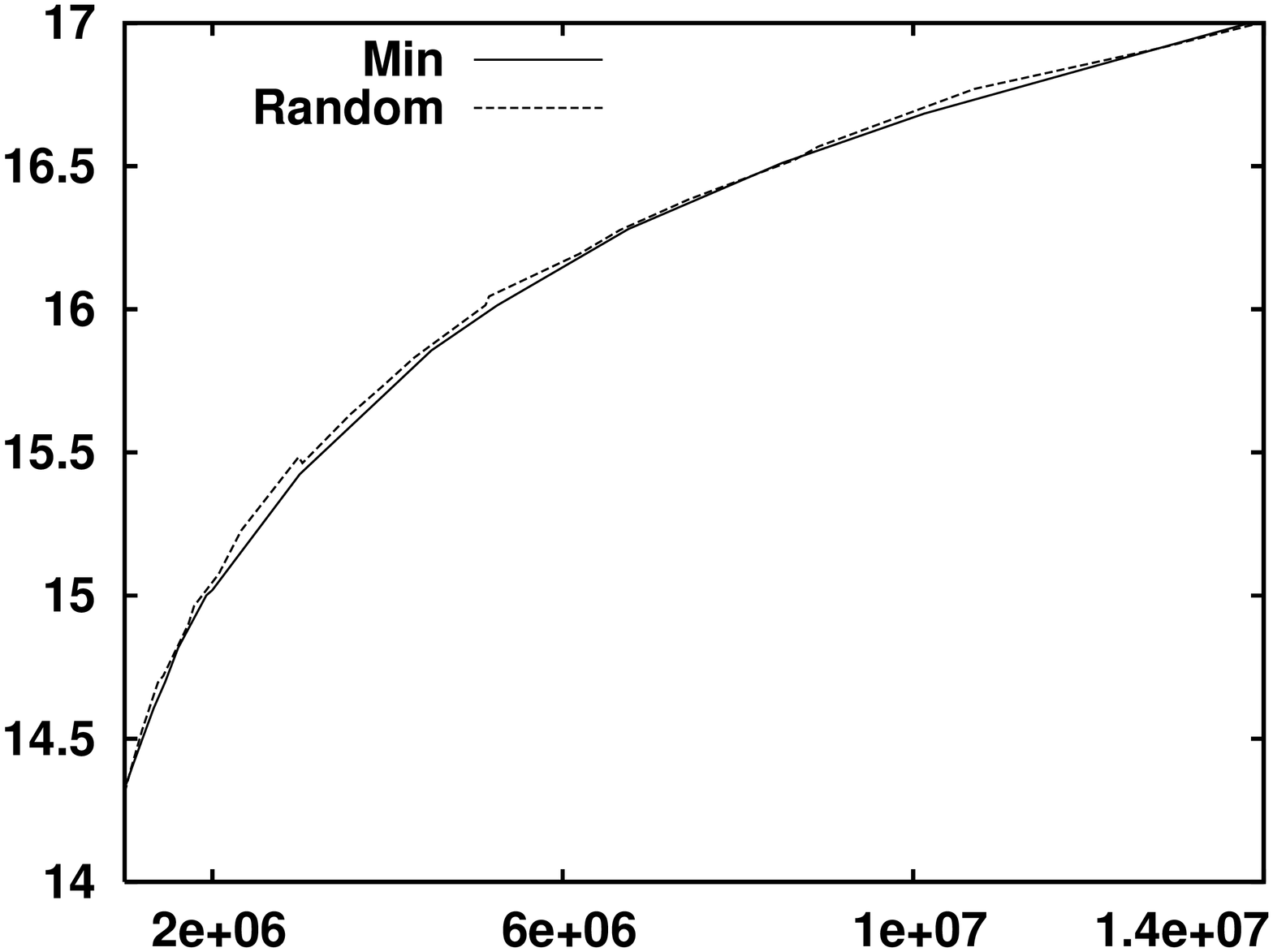}}}
\put(-210,-130){\rotatebox{90}{Time}}
\put(-30,-130){$N$}
\caption{Comparison of Min and Random for $\E(\widetilde T_{\Two})$   when $\rho=2$}\label{fig:24}
\end{center}
\end{figure}

\begin{figure}[ht]
\begin{center}
\rotatebox{-90}{\scalebox{0.3}{\includegraphics{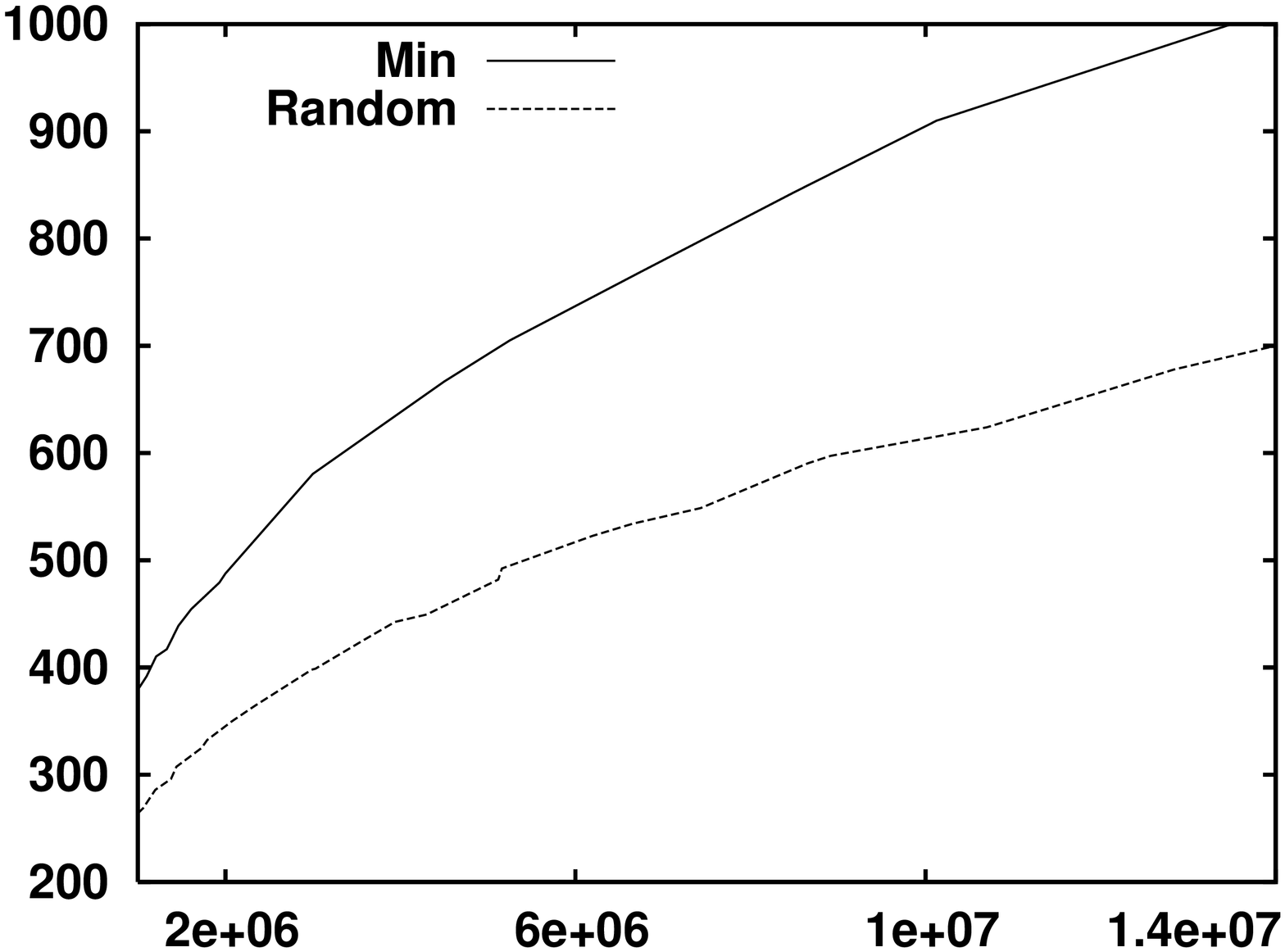}}}
\put(-210,-130){\rotatebox{90}{Time}}
\put(-30,-130){$N$}
\caption{Comparison of Min and Random for $\E(\widetilde \nu_{\Two})$   when $\rho=2$}\label{fig:26}
\end{center}
\end{figure}
Table~\ref{Tab} shows that for the first time when a server becomes empty, the policy has
a great  influence. This is easily understandable: In the min case, it is much harder for a server 
to become empty, because least loaded servers are selected by incoming peers. 

\section{Conclusion}
\label{conclusion}

\begin{comment}
In this paper, we have studied a simple queueing model of a file sharing system. Such a
principle is at the core of a peer-to-peer network. For  this principle, we
have shown the existence  of a first regime during which no  server is idle, thus yielding
the highest possible output rate. In particular, when there are $n$ servers in the system,
the next  server is created at  rate $n$. This observation  makes it possible  to see this
system as an urn and ball problem, where  the process of creation of server forms a random
point process on the half real line
\end{comment}

The simulations moreover underlined the existence of a second regime during which although
 the  fraction of  idle  servers  is small,  the  output rate  is no  longer as   high as
 possible. This second regime is then followed by a third regime during which the capacity
 offered by  the system exceeds by  far the input rate,  and so the  system mainly creates
 empty servers. Our urn and ball approach can no longer be applied to these two regimes,
 and so they will be studied in the near future using other probabilistic techniques.

A possible extension  of our results consists of incorporating the  possibility for a peer
to leave the  system right after completing its  download. In terms of urn  and ball, this
just  amounts to change  the parameter  that defines  the length  of the  $n$-th interval:
instead of $n$, one would just have to  consider $pn$ if $p$ is the probability for a peer
to become  a server after completing  its download.  An  extended model where the  file is
split into  different chunks essentially amounts  to study a  multi-class queueing network
with a  random number of servers  of different classes which proves to be  a much more
difficult problem.

Finally, a natural extension is to consider a general service distribution, instead of the
exponential one.  In this case, the process  of creation of servers can be described as an
age-dependent  branching process,  and more  precisely a  binary  Bellman-Harris branching
process.   See   Athreya~\cite{Athreya77:0,Athreya69:0}.   If  this   setting  complicates
significantly the analysis of  the file sharing system, it seems that  most of the results
obtained in the exponential case should still hold.

\providecommand{\bysame}{\leavevmode\hbox to3em{\hrulefill}\thinspace}
\providecommand{\MR}{\relax\ifhmode\unskip\space\fi MR }
% \MRhref is called by the amsart/book/proc definition of \MR.
\providecommand{\MRhref}[2]{%
  \href{http://www.ams.org/mathscinet-getitem?mr=#1}{#2}
}
\providecommand{\href}[2]{#2}


\begin{thebibliography}{10}

\bibitem{Asmussen:01}
S{\o}ren Asmussen, \emph{Applied probability and queues}, John Wiley \& Sons
  Ltd., Chichester, 1987.

\bibitem{Athreya77:0}
K.~B. Athreya and Niels Keiding, \emph{Estimation theory for continuous-time
  branching processes}, Sankhya: The Indian Journal of Statistics \textbf{89}
  (1977), no.~A, 101--123.

\bibitem{Athreya69:0}
Krishna~B. Athreya, \emph{On the supercritical one dimensional age dependent
  branching processes}, The Annals of Mathematical Statistics \textbf{40}
  (1969), no.~3, 743--763.

\bibitem{Barbour:01}
A.~D. Barbour, Lars Holst, and Svante Janson, \emph{Poisson approximation}, The
  Clarendon Press Oxford University Press, New York, 1992, Oxford Science
  Publications.

\bibitem{Nain}
F.~Cl\'evenot and P.~Nain, \emph{A simple fluid model for the analsysis of the
  {S}quirrel peer-to-peer caching system}, Proc. Infocom, 2004.

\bibitem{Starski}
E.~Cs{\'a}ki and A.~F{\"o}ldes, \emph{On the first empty cell}, Studia
  Scientiarum Mathematicarum Hungarica \textbf{11} (1976), no.~3-4, 373--382
  (1978).

\bibitem{Kurose}
Z.~Ge, D.R. Figueiredo, S.~Jaiswal, J.~Kurose, and D.~Towsley, \emph{Modeling
  peer-to-peer file sharing systems}, Proc. Infocom, 2003.

\bibitem{Gnedin:01}
Alexander Gnedin, Ben Hansen, and Jim Pitman, \emph{Notes on the occupancy
  problem with infinitely many boxes: general asymptotics and power laws},
  Probability Surveys \textbf{4} (2007), 146--171 (electronic). \MR{MR2318403}

\bibitem{Grigelionis61:0}
Bronius Grigelionis, \emph{On the convergence of sums of random step processes
  to a poisson process}, Theory of Probability and its Applications \textbf{8}
  (1961), no.~2, 177--182.

\bibitem{Hwang07:0}
Hsien-Kuei Hwang and Svante Janson, \emph{Local limit theorems for finite and
  infinite urn models}, Annals of Probability (To appear).

\bibitem{Kingman:01}
J.~F.~C. Kingman, \emph{Random partitions in population genetics}, Proceedings
  of the Royal Society. London. Series A. Mathematical, Physical and
  Engineering Sciences \textbf{361} (1978), no.~1704, 1--20.

\bibitem{mathpaper}
Unpublished manuscript~available at, \emph{{\em\small
  http://chambertin.inria.fr/robert/Sigmetrics-Theorem2.pdf}}.

\bibitem{Massoulie}
L.~Massouli\'e and M.~Vojnovi\'{c}, \emph{Coupon replication systems}, Proc.
  Sigmetrics 2005 (Banff, Alberta, Canada), June 2005.

\bibitem{Srikant}
D.~Qiu and R.~Srikant, \emph{Modeling and performance analysis of
  {B}it{T}orrent-like peer-to-peer networks}, Proc. Sigcomm, 2004.

\bibitem{Sanghavi:0}
Sujay Sanghavi, Bruce Hajek, and Laurent Massouli\'e, \emph{Gossiping with
  {M}ultiple {M}essages}, INFOCOM 2007. 26th IEEE International Conference on
  Computer Communications. IEEE, May 2007, pp.~2135--2143.

\bibitem{williams91:0}
David Williams, \emph{Probability with martingales}, Cambridge University
  Press, 1991.

\bibitem{Yang06:0}
Xiangying Yang and Gustavo de~Veciana, \emph{Performance of peer-to-peer
  networks: service capacity and role of resource sharing policies},
  Performance Evaluation \textbf{63} (2006), no.~3, 175--194.

\end{thebibliography}
\end{document}